\documentclass{eptcs}

\providecommand{\event}{MSFP 2014} % Name of the event you are submitting to

\usepackage[all]{xy}
\usepackage{stmaryrd}
\usepackage{amssymb}

\usepackage{proof}

\newtheorem{prop}{Proposition}

\newtheorem{exa}{Example}

\newenvironment{proof}{\noindent \textbf{Proof:}}{\hfill $\Box$ \\ ~}

\newenvironment{proofnoqed}{\noindent \textbf{Proof:}}{}

\newcommand{\inl}{\mathsf{inl}}
\newcommand{\inr}{\mathsf{inr}}

\newcommand{\e}{\mathsf{e}}
\newcommand{\m}{\mathsf{m}}

\newcommand{\JJ}{\mathcal{J}}
\newcommand{\Lan}{\mathrm{Lan}}
\newcommand{\C}{\mathcal{C}}
\newcommand{\op}{\mathrm{op}}
\newcommand{\rev}{\mathrm{rev}}

\newcommand{\Var}{\mathsf{Var}}

\newcommand{\Tm}{\mathsf{Tm}}

\renewcommand{\v}{`}
\newcommand{\I}{\mathsf{I}}
\newcommand{\ot}{\otimes}

\newcommand{\tto}{\Rightarrow}

\newcommand{\al}{\alpha}
\newcommand{\lam}{\lambda}

\newcommand{\id}{\mathsf{id}}
\newcommand{\comp}{\circ}

\newcommand{\Nf}{\mathsf{Nf}}
\newcommand{\J}{\mathsf{J}}
\newcommand{\vot}{\mathbin{`\otimes}}
\newcommand{\bt}{\mathbin{\boxtimes}}

\newcommand{\emb}{\mathsf{emb}}

\newcommand{\nfx}[1]{\llbracket #1 \rrbracket}
\newcommand{\nf}{\mathsf{nf}}

\newcommand{\nmx}[1]{\langle\!\!\langle #1 \rangle\!\!\rangle}
\newcommand{\nm}{\mathsf{nm}}

\newcommand{\rv}{^\mathsf{r}}
\newcommand{\Nfr}{\mathsf{Nf}\rv}
\newcommand{\Jr}{\mathsf{J}\rv}
\newcommand{\votr}{\mathbin{`\otimes}\rv}

\newcommand{\embr}{\mathsf{emb}\rv}

\newcommand{\nfxr}[1]{\llbracket #1 \rrbracket\rv}
\newcommand{\nfr}{\mathsf{nf}\rv}

\newcommand{\nmxr}[1]{\langle\!\!\langle #1 \rangle\!\!\rangle\rv}
\newcommand{\nmr}{\mathsf{nm}\rv}

\newcommand{\CTm}{\textbf{Tm}}
\newcommand{\CNf}{\textbf{Nf}}

\newcommand{\refl}{\textsf{refl}}
\newcommand{\trans}{\textsf{trans}}
\newcommand{\subst}{\textsf{subst}}
\newcommand{\femb}{\textsf{femb}}
\newcommand{\fnfemb}{\textsf{fnfemb}}

\newcommand{\dotminus}{\stackrel{.}{-}}

\begin{document}

\title{Coherence for Skew-Monoidal Categories}
\author{Tarmo Uustalu
\institute{Institute of Cybernetics at Tallinn University of Technology, Estonia}
\email{tarmo@cs.ioc.ee}
}
\def\titlerunning{Coherence for Skew-Monoidal Categories}
\def\authorrunning{T.~Uustalu}

\maketitle

\begin{abstract}
I motivate a variation (due to K.~Szlach\'anyi) of monoidal
categories called skew-monoidal categories where the unital and
associativity laws are not required to be isomorphisms, only natural
transformations. Coherence has to be formulated differently than in
the well-known monoidal case. In my (to my knowledge new) version, it
becomes a statement of uniqueness of normalizing rewrites. I present
a proof of this coherence theorem and also formalize it fully in the
dependently typed programming language Agda.
\end{abstract}

\section{Introduction}

Mac Lane's monoidal categories are ubiquitous in category theory,
mathematics and computer science. One of their remarkable properties
is the coherence theorem stating that, in any monoidal category, any
two parallel maps that are ``formal'' (in the sense that they are put
together from the identity, composition, the tensor, the two unitors
and the associator) are equal. In other words, in the free monoidal
category over a given set of objects, any two maps with the same
domain and codomain are equal. This theorem is both beautiful and
extremely useful. (There is also a simple necessary and sufficient
condition for existence of a map between two given objects in the free
monoidal category.)

Szlach\'anyi \cite{Szl:skemcb} has recently introduced a variation of
monoidal categories, called skew-monoidal categories. The important
difference from monoidal categories is that the unitors and associator
are not required to be isomorphisms. His study was motivated by
structures from quantum physics. In my joint work with Altenkirch and
Chapman \cite{ACU:monnnb}, I ran into the same definition when
generalizing monads to non-endofunctors.

In a free skew-monoidal category over a set of objects, general
uniqueness of parallel maps is lost. But it is still only reasonable
to enquire whether some kind of coherence theorems are possible like
they exist for many types of categories, e.g., Cartesian categories
etc.

In this paper, I state and prove one such theorem. I obtained it by
playing with Beylin and Dybjer's formalization \cite{BD:extpcm} of Mac
Lane's coherence theorem. Essentially, I looked at the high-level
proof structure and checked what can be kept in the skew-monoidal case
and what must necessarily be given up at least if one sticks to the
same overall proof idea. The theorem states that maps to certain
objects---``normal forms''---are unique. As a corollary, the same
holds also for maps from ``reverse-normal forms''. For maps with
different domains and co-domains no information is given.

I have formalized this result in the dependently typed programming
(DTP) language Agda. I found it a very interesting exercise. Of course
this is by no means uncommon with DTP projects, but certainly a
project like this forces one to think carefully about deep matters in
programming with the identity type in intensional type theory.

The structure of this short paper is as follows. I first define
skew-monoidal categories, compare them to monoidal categories, and
give some examples. Then I present the coherence statement and proof
(as formalized in Agda), describe the rewriting intuition behind it
and also hint at what everything means category-theoretically.

The accompanying Agda formalization of the whole development and more
(approximately 750 lines; self-contained, only propositional equality (the
identity type) is taken from the library) is available from 
\url{http://cs.ioc.ee/~tarmo/papers/}.

\section{Skew-monoidal categories}

Skew-monoidal categories of Szlach\'anyi \cite{Szl:skemcb} are a
variation of monoidal categories, originally due to Mac Lane
\cite{ML:natac}.

A \emph{(left) skew-monoidal category} is a category $\C$ together
with a distinguished object $\I$, a functor $\ot : \C \times \C \to
\C$ and three natural transformations $\lam$, $\rho$, $\al$ typed 
\[
\begin{array}{c}
\lam_{A} : \I \ot A \to A \\
\rho_{A} : A \to A \ot \I\\
\al_{A,B,C} : (A \ot B) \ot C \to A \ot (B \ot C)
\end{array}
\]
satisfying the laws
\[
\mathrm{(a)}  
\xymatrix@R=1.5pc@C=0.2pc{
    & \I \ot \I \ar[dr]^-{\lambda_\I} & \\
    \I \ar[ur]^-{\rho_\I} \ar@{=}[rr] & & \I
    }
\quad
\mathrm{(b)}  
\xymatrix@R=1.3pc{
      (A\ot \I) \ot B \ar[r]^{\alpha_{A,\I,B}}
      & A \ot (\I\ot B) \ar[d]^{A\ot \lambda_{B}}\\
      A \ot B \ar@{=}[r] \ar[u]^{\rho_{A}\ot B}&  A \ot B 
    }
\]
\[
\mathrm{(c)}  
\xymatrix@C=0.2pc@R=1.5pc{
  (\I \ot A) \ot B \ar[dr]_{\lambda_A \ot B} \ar[rr]^{\alpha_{\I,A,B}} 
           & &  \I \ot (A \ot B) \ar[dl]^{\lambda_{A \ot B}}\\
  & A \ot B & 
    }
\quad
\mathrm{(d)}  
\xymatrix@C=0.2pc@R=1.5pc{
  (A \ot B) \ot \I \ar[rr]^{\alpha_{A,B,\I}} 
           & &  A \ot (B \ot \I) \\
  & A \ot B \ar[ul]^{\rho_{A \ot B}} \ar[ur]_{A \ot \rho_B} & 
    }
\]
\[
\mathrm{(e)} 
\xymatrix@R=1.5pc@C=2.5pc{
(A\ot (B \ot C)) \ot D \ar[rr]^{\alpha_{A,B \ot C,D}}
  & & A\ot ((B \ot C)\ot D) \ar[d]^{A\ot \alpha_{B,C,D}}
  \\
((A\ot B) \ot C) \ot D \ar[u]^{\alpha_{A,B,C} \ot D}
      \ar[r]^{\alpha_{A\ot B,C,D}}
  & (A\ot B) \ot (C \ot D) \ar[r]^{\alpha_{A,B,C\ot D}}
    & A\ot (B \ot (C \ot D))
}
\]
A monoidal category is obtained, if $\lam$, $\rho$, $\al$ are
additionally required to be isomorphisms. Here this requirement is not
made.

When dropping the requirement of isomorphisms from the definition of
monoidal categories, the key question is of course which half of each
of the three isomorphisms should be kept and how the laws (coherence
conditions) should be stated. In a left skew-monoidal category $\lam$
``reduces'', $\rho$ ``expands'' and $\al$ ``associates to the
right''. With this decision, the monoidal category laws (c), (d), (e)
can be stated in only one way. But for (a) there are two choices
($\lam_\I \comp \rho_\I = \id_\I$ and $\rho_\I \comp \lam_\I = \id_{\I
  \ot \I}$) and for (b) even three. The ``correct'' options turn to be
those we have chosen.

Notice also that (a-e) are directed versions of the original Mac Lane
axioms \cite{ML:natac}. Later Kelly \cite{Kel:maclcc} discovered that
(a), (c), (d) can be derived from (b) and (e). For skew-monoidal
categories, this is not the case!

There is also an analogous notion of a \emph{right skew-monoidal
  category}. It is important to realize that the opposite category
$\C^\op$ of a left skew-monoidal category $\C$ is right skew-monoidal,
not left-skew monoidal. But the ``reverse'' version $(\C^\op)^\rev$ of
$\C^\op$ (obtained by swapping the arguments of $\otimes$ and also
$\lam$ and $\rho$) is left skew-monoidal.

In the rest of this text, my focus is on left skew-monoidal
categories. Speaking of a skew-monoidal category without specifying
its skew, I mean the left skew.

Here are some examples of skew-monoidal categories.

\begin{exa}
A simple example of a skew-monoidal category resulting from
skewing a numeric addition monoid is the following.

View the partial order $(\mathbb{N}, \leq)$ of natural numbers as a
thin category. Fix some natural number $n$ and define $\I = n$ and $x
\otimes y = (x \dotminus n) + y$ where $\dotminus$ is ``truncating
subtraction''. We have $\lambda_x : (n \dotminus n) + x = 0 + x = x$,
$\rho_x : x \leq x \mathbin{\max} n = (x \dotminus n) + n$,
$\alpha_{x,y,z} : (((x \dotminus n) + y) \dotminus n) + z \leq (x
\dotminus n) + (y \dotminus n) + z$ (by a small case analysis).
\end{exa}

\begin{exa}
The category of pointed sets and point-preserving functions has the
following skew-monoidal structure. 

Take $\I = (1, \ast)$ and $(X, p) \otimes (Y, q) = (X + Y, \inl~p)$
(notice the ``skew'' in choosing the point). We define $\lam_X : (1,
\ast) \ot (X, p) = (1+X, \inl~\ast) \to (X, p)$ by $\lam_X~(\inl~\ast)
= p$, $\lam_X~(\inr~x) = x$ (this is not injective).  We let $\rho_X :
(X, p) \to (X + 1, p) =(X, \inl~p) \ot (1, \ast)$ by $\rho_X~x =
\inl~x$ (this is not surjective). Finally we let $\al_{X,Y,Z} : ((X, p)
\ot (Y, q)) \ot (Z, r) = ((X + Y) + Z, \inl~(\inl~p)) \to (X + (Y + Z),
\inl~p) = (X, p) \ot ((Y, q) \ot (Z, r))$ be the obvious isomorphism.
\end{exa}

\begin{exa}
Given a monoidal category $(\C, \I, \ot)$. Given also a lax
monoidal comonad $(D, \e, \m)$ on $\C$. The category $\C$ has a
skew-monoidal structure given by $\I^D = \I$, $A \ot^D B = A \ot
D~B$. The unitors and associator are the following:
\[
\begin{array}{c}
\lam^D_A = \xymatrix{\I \ot D~A \ar[r]^-{\lam_{DA}} & D~A \ar[r]^-{\epsilon_A} & A} \\
\rho^D_A = \xymatrix{A \ar[r]^-{\rho_A} & A \ot \I \ar[r]^-{A \ot \e} & A \ot D~\I} \\
\hspace*{-3mm}
\al^D_{A,B,C} = \xymatrix@C=2pc{(A \ot D~B) \ot D~C 
                   \ar[r]^-{(A \ot DB) \ot \delta_C}
                 & (A \ot D~B) \ot D~(D~C) \ar[r]^-{\al_{A,DB,DC}}
                 & A \ot (D~B \ot D~(D~C)) \ar[r]^-{A \ot \m_{B,C}}
                 & A \ot D~(B \ot D~C)}
\end{array}
\]
A similar skew-monoidal category is also obtained with an oplax
monoidal monad.
\end{exa}

\begin{exa}
Consider two categories $\JJ$ and $\C$ and a functor $J : \JJ \to \C$.
The functor category $[\JJ, \C]$ has a skew-monoidal structure given
by $\I = J$, $F \otimes G = \mathrm{Lan}_J F \cdot G$ (assuming that
the left Kan extension $\Lan_J\, F : \C \to \C$ exists for every $F : \JJ
\to \C$).  The unitors and associator are the canonical natural
transformations $\lam_F : \Lan_J~J \cdot F \to F$, $\rho_F : F \to
\Lan_J~F \cdot J$, $\al_{F,G,H} : \Lan_J~(\Lan_J~F \cdot G) \cdot H
\to \Lan_J~F \cdot \Lan_J~G \cdot H$. This category becomes properly
monoidal under certain conditions on $J$. $\rho$ is an isomorphism, if
$J$ is fully-faithful. $\lam$ is an isomorphism, if $J$ is
dense. (This is the example from our relative monads work
\cite{ACU:monnnb}. Relative monads on $J$ are skew-monoids in the
skew-monoidal category $[\JJ, \C]$.)
\end{exa}

\section{The coherence theorem}

I now give a sufficient criterion for equality of two parallel maps in
the free skew-monoidal category.

I first present the minimal technical development leading to a
statement and proof of the result, not commenting at all on what
everything means category-theoretically. (This development follows the
Agda formalization.) Then I give a rewriting ``interpretation'' of the
story. Finally I explain the categorical meaning of the result.

The objects of the free skew-monoidal category over a set $\Var$
of objects are given by the set of ``object expressions'' $\Tm$
defined inductively as follows:
\[
\infer{\v X : \Tm}{X : \Var}
\quad
\infer{\I : \Tm}{}
\quad
\infer{A \ot B : \Tm}{A : \Tm & B : \Tm}
\]

The maps between two objects $A$ and $B$ are given by the set $A \tto
B$ of ``map expressions'' quotiented by the relation $\doteq$ of
``derivable equality''. The former is defined inductively by the rules
\[
\renewcommand{\arraystretch}{2}
\begin{array}{c}
\infer{\id : A \tto A}{}
\quad
\infer{f \comp g : A \tto C}{f : B \tto C & g : A \tto B}
\\
\infer{f \ot g : A \ot B \tto C \ot D}{f : A \tto C & g : B \tto D}
\\
\infer{\lam : I \ot A \tto A}{}
\quad
\infer{\rho : A \tto A \ot I}{}
\quad
\infer{\al : (A \ot B) \ot C \tto A \ot (B \ot C)}{}
\end{array}
\]
while the latter is defined inductively by the rules
\[
\renewcommand{\arraystretch}{2}
\begin{array}{c}
\infer{f \doteq f}{}
\quad
\infer{g \doteq f}{f \doteq g}
\quad
\infer{f \doteq h}{f \doteq g & g \doteq h}
\quad
\infer{f \comp h \doteq g \comp k}{f \doteq g & h \doteq k}
\quad
\infer{f \ot h \doteq g \ot k}{f \doteq g & h \doteq k}
\\
\infer{\id \comp f \doteq f}{}
\quad
\infer{f \doteq f \comp \id}{}
\quad
\infer{(f \comp g) \comp h \doteq f \comp (g \comp h)}{}
\\
\infer{\id \ot \id \doteq \id}{}
\quad 
\infer{(h \comp f) \ot (k \comp g) \doteq h \ot k \comp f \ot g}{}
\\
\infer{\lam \comp \id \ot f \doteq f \comp \lam}{}
\quad
\infer{\rho \comp f \doteq f \ot \id \comp \rho}{}
\quad
\infer{\al \comp (f \ot g) \ot h \doteq f \ot (g \ot h) \comp \al}{}
\\
\infer{\lam \comp \rho \doteq \id}{}
\quad
\infer{\id \doteq \id \ot \lam \comp \al \comp \rho \ot \id}{}
\\
\infer{\lam \comp \al \doteq \lam \ot \id}{}
\quad
\infer{\al \comp \rho \doteq \id \ot \rho}{}
\quad
\infer{\al \comp \al \doteq \id \ot \al \comp \al \comp \al \ot \id}{}
\end{array}
\]

We define ``normal forms'' of object expressions as the set $\Nf$
defined inductively by
\[
\infer{\J : \Nf}{}
\quad
\infer{X \vot N : \Nf}{X : \Var & N : \Nf}
\]
Normal forms embed into object expressions via the function $\emb :
\Nf \to \Tm$ defined recursively by
\[
\begin{array}{l}
\emb~\J = \I \\
\emb~(X \vot N) = \v X \ot \emb~N
\end{array}
\]

Let $\nfx{-} : \Tm \to \Nf \to \Nf$ be the function defined
recursively by the element of $\Tm$ by
\[
\begin{array}{l}
\nfx{\v X}~N = X \vot N \\
\nfx{\I}~N = N \\
\nfx{A \otimes B}~N = \nfx{A}~(\nfx{B}~N)
\end{array}
\]
Every object expression is assigned a normal form with the
normalization function $\nf : \Tm \to \Nf$ defined by
\[
\nf~ A = \nfx{A}~ \J
\]

We can make some first important observations.

\begin{prop}\label{fnf}~ %\\
\begin{enumerate}
\item For any $f : A \tto B$ and $N : \Nf$, $\nfx{A}~N = \nfx{B}~N$.
\item For any $f : A \tto B$, $\nf~ A = \nf~ B$.
\end{enumerate}
\end{prop}

\begin{proofnoqed}
\begin{enumerate}
\item By induction on $f$. 
\item Immediate from (1).
\hfill$\Box$
\end{enumerate}
\end{proofnoqed}

\begin{prop}\label{nfemb}
For any $N : \Nf$, $\nf~(\emb~ N) = N$.
\end{prop}

\begin{proof} By induction on $N$.
\end{proof}

\begin{prop}\label{fnfemb}
For any $f : A \tto \emb~N$, $\nf~A = N$.
\end{prop}

\begin{proof}
An immediate combination of Propositions \ref{fnf}(2) and \ref{nfemb}.
\end{proof}

Let now $\nmx{-} : \Pi A : \Tm.~ \Pi N : \Nf.~ A \ot \emb~N \tto
\emb~(\nfx{A}~N)$ be the function defined by
\[
\begin{array}{l}
\nmx{\v X}~N = \id \\
\nmx{\I}~N = \lam \\
\nmx{A \ot B}~N = \nmx{A}~(\nfx{B}~N) \comp \id \ot \nmx{B}~N \comp \al
\end{array}
\]
To every object expression we assign a ``normalizing'' map expression
with the function $\nm : \linebreak \Pi A: \Tm.~A \tto \emb~(\nf~A)$ defined by
\[
\nm~A = \nmx{A}~\J \comp \rho
\]

We are ready to state our result.

\begin{prop}[Main lemma]\label{huh} %~\\
\begin{enumerate}
\item For any $f : A \tto B$ and $N : \Nf$, $\nmx{A}~ N \doteq
  \nmx{B}~N \comp f \ot \id$. (This statement is well-formed, as
  $\nfx{A}~N = \nfx{B}~N$ by Proposition \ref{fnf}(1).)
\item For any $f : A \tto B$, $\nm~A \doteq \nm~B \comp f$. (This
statement is well-formed, as $\nf~A = \nf~B$ by Proposition \ref{fnf}(2).)
\end{enumerate}
\end{prop}

\begin{proofnoqed} 
\begin{enumerate}
\item By induction on $f$. This is a tedious but simple proof with six
  cases, some are tricky for formalization! Read the Agda
  development.

  Of course the proof relies on the equality of map expressions being
  induced by the five coherence conditions. All of them are needed and
  exactly in the versions chosen (remember that for conditions (a),
  (b) there were multiple inequivalent options).

\item Follows from (1).
\hfill $\Box$
\end{enumerate}
\end{proofnoqed}

\begin{prop}\label{hei}
For any $N : \Nf$, $\nm~(\emb~N) \doteq \id$. (This statement
is well-formed, as  $\nf~(\emb~N) = N$ by Proposition \ref{nfemb}.)
\end{prop}

\begin{proof}
By induction on $N$.
\end{proof}

\begin{prop}[Main theorem] \label{nnmemb}
For any $f : A \tto \emb~N$, $\nm~A \doteq f$. (This statement
is well-formed, because $\nf~A = N$ by Proposition \ref{fnfemb}).
\end{prop}

\begin{proof}
By combining Propositions \ref{huh}(2) and \ref{hei}.
\end{proof}

Of course nothing prevents us from playing the reverse game. We can
define a set $\Nfr$ and functions $\embr : \Nfr \to \Tm$, $\nfxr{-} :
\Tm \to \Nfr \to \Nfr$ and $\nfr : \Tm \to \Nfr$:
\[
\infer{\Jr : \Nfr}{}
\quad
\infer{R \votr X : \Nfr}{X : \Var & R : \Nfr}
\]
\[
\begin{array}{l}
\embr~\Jr = \I \\
\embr~(R \votr X) = \embr~R \ot \v X \\[1em]
\nfxr{\v X}~R = R \votr X \\
\nfxr{\I}~R = R \\
\nfxr{A \otimes B}~R = \nfxr{B}~(\nfxr{A}~R) \\[1em]
\nfr~ A = \nfxr{A}~ \Jr
\end{array}
\]

Further we can define functions $\nmxr{-} : \Pi A : \Tm.~ \Pi R :
\Nfr.~\embr~(\nfxr{A}~R) \tto \embr~R \ot A$ and $\nmr : \Pi A:
\Tm.~\embr~(\nfr~A) \tto A$ and propositions as above hold for
them. Furthermore, $\nf~A = \nf~B$ if and only if $\nfr~A= \nfr~B$.

\medskip

Thus we see that for two object expressions $A$ and $B$ to have
exactly one map expression between them (up to $\doteq$), it suffices to
have $A = \embr~R$ or $B = \emb~N$ for some $R$ or $N$ (i.e., $A$ in
reverse normal form or $B$ in normal form).

It is important to notice that this is merely a sufficient condition
for a unique map expression between two object expressions. It is
perfectly possible have a unique map expression between $A$ and $B$
even if $A$ is not a reverse normal form and $B$ is not a normal
form. The simplest example is $A = B = \v X$, since we have only the
map $\id : \v X \tto \v X$.

At the same time it is easy to find pairs of object expressions $A$
and $B$ with $\nf~A = \nf~B$ (which is the same as $\nfr~A= \nfr~B$)
with no or several map expressions between them.

Some examples of absence of map expressions: 
\begin{itemize}
\item there is no map exression from $\v X$ to $\I \ot \v X$
  (although both have $X \vot \J$ as the normal form);
\item there is no map expression from $\v X \ot ((\v Y \ot \v Z) \ot
  \I)$ to $(\v X \ot \v Y) \ot (\v Z \ot \I)$ (despite both having $X
  \vot (Y \vot (Z \vot \J))$ as the normal form).
\end{itemize}

Some examples of multiple map expressions: 
\begin{itemize}
\item $\id \not\doteq \rho \comp \lam  : \I \ot \I \tto \I \ot \I$, 
\item $\id \not\doteq \rho \ot \id \comp \id \ot \lam \comp \al : (\v X \ot \I) \ot \v Y \tto (\v X \ot \I) \ot \v Y$,
\item $\lam \not\doteq \id \ot \lam : \I \ot (\I \ot \v X) \tto \I \ot \v X$.
\end{itemize}

\paragraph{Rewriting interpretation}

Let us see what we have established from the rewriting perspective.

The elements of $\Tm$ can be thought of as \emph{terms} over $\Var$
made of $\I$ and $\otimes$. The elements of $A \tto B$ should be
thought of as \emph{rewrites} of $A$ into $B$: $\lambda$, $\rho$,
$\alpha$ are rewrite rules, $\otimes$ allows applying rewrite rules
inside a term, $\id$ is the nil rewrite, $\comp$ is sequential
composition of two rewrites. The relation $\doteq$ provides a
congruence on rewrites of one term to another.

With $\Nf$ we have carved out from the set of all terms $\Tm$ some
terms that we have decided to consider to be in \emph{normal
  form}. $\emb$ is the inclusion of this set $\Nf$ into $\Tm$.

With $\nf~A$ we have assigned to every term a particular normal form,
which we define to be its (unique) normal form. Notice that this is an
entirely rewriting-independent definition of normalization.

Proposition 1 says that one term can only be written into another, if their
normal forms are the same. Proposition 2 says that a normal form's normal
form is itself. Proposition 3 is the obvious conclusion that, if a term
rewrites to a normal form, it is that term's normal form.

With $\nm~A$ we have at least one (canonical) rewrite of any term $A$
to its normal form.

Proposition 4 says that the canonical normalizing rewrite $\nm~A$ of a term
$A$ factors through all other rewrites of it. Proposition 5 says that
the canonical normalizing rewrite of a normal form (into itself) is
the nil rewrite.

Proposition 6 tells us that any normalizing rewrite of a term $A$ is equal
to the canonical normalizing rewrite $\nm~ A$. Thus all normalizing
rewrites of $A$ are equal.

\paragraph{Categorical meaning}

Categorically speaking we have established a relationship between
two categories $\CTm$, which is the free skew-monoidal category over
$\Var$, and $\CNf$, which is the free strictly monoidal category over
$\Var$.

The category $\CTm$ has $\Tm$ as the set of objects, $(A \tto
B)/\doteq$ the set of maps between objects $A$, $B$, $\id$ the
identity, $\circ$ composition, $I$ as the unit, $\ot$ the tensor,
$\lam$, $\rho$ and $\al$ the unitors and associator. The category
$\CNf$ is discrete and has $\Nf$ (the set of lists over $\Var$) as the
set of objects. 

$\emb$ with the trivial map mapping is clearly a functor from $\CNf$
to $\CTm$.

With $\nf : \Tm \to \Nf$ we provide an object mapping for a functor
from $\CTm$ to $\CNf$. Proposition 1, stating that, for any $f : A \to B$,
we have $\nf~A = \nf~B$, tells us that $\nf$ with the constant
identity map mapping is a functor from $\CTm$ to $\CNf$ (``if $f : A
\to B$, then $L~f : L~A \to L~B$'').

Proposition 2, stating that $\nf~(\emb~N) = N$, establishes that identity
has the correct type for being the counit, if $\nf$ were a left
adjoint of $\emb$ (``$\epsilon : L~(R~N) \to N$'').

Proposition 3 concludes from Propositions 1 and 2 that $f : A \tto \emb~N$
implies $\nf~A = N$. This shows that constant identity is a good
candidate for the left transpose operation of such an adjunction (``if
$f : A \to R~N$, then $f^\dagger : L~A \to N$''). The proof is nothing
but the standard definition of left transpose from the counit
(``$f^\dagger = \epsilon \comp L f$'').

The polymorphic function $\nm : A \tto \emb~(\nf~A)$ is,
by its typing, a candidate for the unit of the adjunction (``$\eta : A
\to R~(L~A)$''). The main lemma (Proposition 4), stating that, if $f : A
\tto B$, then $\nm~A \doteq \nm~B \comp f$, establishes that $\nm$ is
a natural transformation (``$R~(L~f) \comp \eta = \eta \comp f$'').

Proposition 5, stating $\nm~(\emb~N) = \id$, establishes one of the
adjunction laws (``$R~\epsilon \comp \eta = \id$'').

The main theorem (Proposition ), which is a conclusion from
Propositions 4 and 5 and states that $f : A \tto \emb~N$ implies
$\nm~A = f$, is a proof of the equivalent adjunction law in terms of
the left transpose (``$R f^\dagger \comp \eta = f$'').

The other adjunction laws are trivial as $\CNf$ is a discrete
category.  Hence our coherence result really establishes that
$\nf$ and $\emb$ provide an adjunction between $\CTm$ and $\CNf$.

\bigskip

In fact, much more can be proved. As already mentioned above, the
category $\CNf$ is strictly monoidal. The unit is $\J$ (the empty
list) and the tensor is $\bt$ (concatenation of lists). More, it is
the free strictly monoidal category over $\Var$. The functors $\nf :
\CTm \to \CNf$ and $\emb : \CNf \to \CTm$ are lax
skew-monoidal. Further, the unit and (trivially) the counit are lax
skew-monoidal too, so the adjunction between $\CTm$ and $\Nf$ is a lax
skew-monoidal adjuction. The Agda development has the full proofs.

Also, I have downplayed here a further fact (which the proof however
relies on implicitly) that the adjunction between $\CTm$ and $\CNf$
factors through the discrete functor category $[\CNf,\CNf]$.

\paragraph{Equality of normal forms in Agda}

My development makes heavy use of equality reasoning on normal
forms. Notice in particular that the statements of Propositions 4--6 are
well-formed only since Propositions 1-3 (stating equalities of normal
forms) hold.

In the Agda formalization, based on intensional type theory, I model
equality of normal forms with propositional equality (the identity
type) on $\Nf$. In the sense of the Agda development, the maps of
$\CNf$ are exactly proofs of equality propositions $N \equiv
N'$. In particular, identity is $\refl$ and composition is
$\trans$. Discreteness is the uniqueness of identity proofs principle.

A consequence of this is that, if $N$ and $N'$ are equal only
propositionally with a proof $p$, and not definitionally, then we
cannot form the proposition $f \doteq g$ for two map expressions $f : A \tto
\emb~N$ and $g : A \tto \emb~N'$. What is well-formed is $\subst~(A \tto
  \emb~{-})~p f \doteq g$. For example, one Agda formulation of Proposition 6
    could be
\[
\Pi f : A \tto \emb~N.~\subst~(A \tto \emb~{-})~(\fnfemb~f)~(\nm~A) \doteq f
\]
where $\fnfemb$ is the proof of Proposition 3, i.e., of $\Pi f : A \tto
\emb~N.~\nf~A \equiv N$, the left transpose operation of the
adjunction between $\nf$ and $\emb$.

Working with $\subst$ (or alternatives, like \texttt{with} and
pattern-matching on propositional equality proofs, or
\texttt{rewrite}) is tedious.

I was therefore relieved to find that, for this project, there is a
neat alternative. Substitution for $N$ in the type $A \tto \emb~N$ of a
map expression $f$ based on $p : N \equiv N'$ can be replaced by
postcomposition of $f$ with $\femb~p: \emb~N \to \emb~N'$ where 
\[
\femb~p= \subst~(\emb~N \tto \emb~{-})~p~(\id~\{\emb~N\})
\] 
is the identity map expression on $\emb~N$ with its codomain
``adjusted'' to $\emb~N'$.

For any $f : A \tto \emb~N$ and $p : N \equiv N'$, it is the case that
$\femb~p \comp f \doteq \subst~(A \tto \emb~{-})~p~f$.

In particular, Proposition 6 can say,
\[
\Pi f : A \tto \emb~N.~\femb~(\fnfemb~f) \circ \nm~A \doteq f
\]

An additional advantage is that the map mapping part of the functor
$\emb$ from $\CNf$ to $\CTm$, which is otherwise obscured in the Agda
formalization, becomes manifest. For example, Proposition 6 explicitly
obtains the form ``$R f^\dagger \comp \eta = f$'', in which we
immediately recognize one of the adjunction laws.

\section{Related work}

Skew-monoidal categories were first studied as such by Szlach\'anyi
\cite{Szl:skemcb} in the context of structures for quantum
computing. They immediately attracted the interest of Lack, Street,
Buckley and Garner \cite{LS:skemsw,LS:triosm,BGLS:skemcc}. Lack and
Street \cite{LS:triosm} proved a coherence theorem, which is different
from the one here: they give a necessary and sufficient condition for
equality two parallel maps of the free skew-monoidal category.

I first met a skew-monoidal category in my work with Altenkirch and
Chapman \cite{ACU:monnnb} on relative monads: we noticed and made use
of the skew-monoidal structure (non-endo)-functor categories. The
context categories of Blute, Cockett and Seely \cite{BCS:catccu} have
the skew-monoidal category data and laws as part of the structure.

Some other weakened versions of monoidal categories with the unitors
and associator not isomorphisms are the pseudocategories of Burroni
\cite{Bur:tcat} (like left skew-monoidal categories but with $\lam$
``expanding'') and Grandis d-lax 2-categories \cite{Gra:laxtcd} ($\al$
``associating to the left'').

Laplaza \cite{Lap:cohani} studied coherence for a version of
semimonoidal categories with associativity not an isomorphism in the
1970s.

Coherence is generally related to equational reasoning and through
that to term rewriting. Beke \cite{Bek:cattrk} considered replacing
the question of uniqueness of equality proofs (equality of parallel
maps in categories with isomorphisms only) with uniqueness of
(normalizing) rewrites of maps (equality of parallel maps in a more
general setting). He also asked whether coherence could be proved for
structures like skew-monoidal categories.

\section{Conclusion and future work}

My main conclusion is the same as Beke's \cite{Bek:cattrk}: for
skew-structured categories (with natural transformations instead of
natural isomorphisms), coherence is not about uniqueness of equality
proofs in an equational theory, but about uniqueness of rewriting
(typically normalization) proofs in a rewrite system. A skew coherence
theorem can give a new insight into the proof of the corresponding
non-skew theorem. In my case, I realized I could heavily build on the
proof of Beylin and Dybjer \cite{BD:extpcm} of coherence for monoidal
categories. It literally felt that all of my proof was already present
in theirs, only the theorem was missing!

My next goal is to formulate and prove a similar coherence theorem for
skew-closed categories of Street \cite{Str:sekcc}, a skew version of
Eilenberg and Kelly's (non-monoidal) closed categories
\cite{EK:cloc}. A coherence theorem for closed categories was obtained
by Laplaza \cite{Lap:cohnmc}. Like in this paper, I aim at a proof
formalized in Agda. Coherence proofs tend to have a conceptually
interesting high-level structure, but underneath they involve many
tedious uninspiring case distinctions; it is more than easy to make
mistakes.

\paragraph{Acknowledgements}
I thank my anonymous reviewers for their useful comments.

This research was supported by the ERDF funded Estonian CoE project
EXCS and ICT National Programme project ``Coinduction'', the Estonian
Science Foundation grant no.~9475 and the Estonian Ministry of
Education and Research target-financed research theme no.~0140007s12.

\nocite{*}

\bibliographystyle{eptcs}
%\bibliography{generic}

\begin{thebibliography}{99}

\bibitem{ACU:monnnb} T. Altenkirch, J. Chapman \& T. Uustalu (2010):
  \emph{Monads need not be endofunctors}. In L. Ong, ed.:
  \textsl{Proc. of 13th Int. Conf. on Foundations of Software Science
    and Computation Structures, FoSSaCS 2010 (Paphos, March 2010)},
  \textsl{Lect. Notes in Comput. Sci.} 6014, Springer, 2010,
  pp.~297--311,  \doi{10.1007/978-3-642-12032-9\_21}.

\bibitem{Bek:cattrk} T. Beke (2011): \emph{Categorification, term
  rewriting and the Knuth-Bendix procedure}. \textsl{J. of Pure and
  Appl. Alg.} 215(5), pp.~728--740, \doi{10.1016/j.jpaa.2010.06.019}.

\bibitem{BD:extpcm} I. Beylin \& P. Dybjer (1996): \emph{Extracting a
  proof of coherence for monoidal categories from a proof of
  normalization for monoids}. In S. Berardi \& M. Coppo, eds.:
  \textsl{Selected Papers from Int.\ Wksh.\ on Types for Proofs and
    Programs, TYPES~'95 (Torino, June 1995)}, \textsl{Lect. Notes in
    Comput. Sci.}  1158, Springer, Berlin,
  pp.~47--61, \doi{10.1007/3-540-61780-9\_61}.

\bibitem{BCS:catccu} R. F. Blute, J. R. G. B. Cockett \&
  R. A. G. Seely (1997): \emph{Categories for computation in context and
    unified logic}. \textsl{J. of Pure and Appl. Alg.} 116(1--3),
  pp.~49--98, \doi{10.1016/s0022-4049(96)00162-4}.

\bibitem{BGLS:skemcc} M. Buckley, R. Garner, S. Lack \& R. Street
  (2013): \emph{Skew-monoidal categories and the Catalan simplicial
  set}. arXiv preprint arXiv:1307.0265. Available at
  \url{http://arxiv.org/abs/1307.0265}.

\bibitem{Bur:tcat} A. Burroni (1971): \emph{T-cat\'egories
  (cat\'egories dans un triple)}. \textsl{Cahiers de topologie et
  g\'eom\'etrie diff\'erentielle cat\'egoriques} 12(3), pp.~215--321.
  Available at \url{http://www.numdam.org/item?id=CTGDC\_1971\_\_12\_3\_215\_0}.

\bibitem{EK:cloc} S. Eilenberg \& M. Kelly (1966): \emph{Closed
  categories}. In S. Eilenberg, D. K. Harrison, S. Mac Lane \&
  H. R\"ohrl, eds.: \textsl{Proc. of Conf. on Categorical Algebra (La
    Jolla, 1965)},
  pp.~421--562, \doi{10.1007/978-3-642-99902-4\_22}.

\bibitem{Gra:laxtcd} M. Grandis (2006): \emph{Lax 2-categories and
  directed homotopy}. \textsl{Cahiers de topologie et g\'eom\'etrie
  diff\'erentielle cat\'egoriques} 47(2), pp.~107--128.
  Available at \url{http://www.numdam.org/item?id=CTGDC\_2006\_\_47\_2\_107\_0}.

\bibitem{Kel:maclcc} G. M. Kelly (1964): \emph{On MacLane's
  conditions for coherence of natural associativities,
  commutativities, etc.}  \textsl{J. of Alg.}  1(4),
  pp. 397--402, \doi{10.1016/0021-8693(64)90018-3}.

\bibitem{LS:skemsw} S. Lack \& R. Street (2012): \emph{Skew
  monoidales, skew warpings and quantum categories}. \textsl{Theory
  and Appl. of Categ.} 26, pp.~385--402. Available at
  \url{http://www.tac.mta.ca/tac/volumes/26/15/26-15abs.html}.

\bibitem{LS:triosm} S. Lack \& R. Street (2013): \emph{Triangulations,
  orientals, and skew monoidal categories}. arXiv preprint
  arXiv:1302.4488. Available at \url{http://arxiv.org/abs/1302.4488}.

\bibitem{Lap:cohani} M. L. Laplaza (1972): \emph{Coherence for
  associativity not an isomorphism}. \textsl{J. of Pure and Appl. Alg.}
  2(2), pp. 107--120, \doi{10.1016/0022-4049(72)90016-3}.

\bibitem{Lap:cohnmc} M. L. Laplaza (1977): \emph{Coherence in
  nonmonoidal closed categories}. \textsl{Trans. of Amer. Math. Soc.}
  230, pp. 293--311, \doi{10.1090/s0002-9947-1977-0444740-9}.

\bibitem{ML:natac} S. Mac Lane (1963): \emph{Natural associativity and
  commutativity}. \textsl{Rice Univ. Stud.} 49(4), pp.~28--46.
  Available at \url{http://hdl.handle.net/1911/62865}.

\bibitem{Str:sekcc} R. Street (2013): \emph{Skew-closed
  categories}. \textsl{J. of Pure and Appl. Alg.} 217(6),
  pp.~973--988, \doi{10.1016/j.jpaa.2012.09.020}.

\bibitem{Szl:skemcb} K. Szlach\'anyi (2012): \emph{Skew-monoidal
  categories and bialgebroids}. \textsl{Adv. in Math.} 231(3--4),
  pp.~1694--1730, \doi{10.1016/j.aim.2012.06.027}.

\end{thebibliography}

\newcommand{\doi}[1]{\href{http://dx.doi.org/#1}{doi: #1}}

\end{document}